\documentclass[3p, review]{elsarticle}
\usepackage{amsmath}
\usepackage{amsthm}
\usepackage{amssymb}

\usepackage{subfig}
\usepackage{graphicx}
\usepackage{booktabs}
\usepackage{braket}
\usepackage{enumerate}
\usepackage[title]{appendix}
\usepackage[colorlinks=true, allcolors=blue]{hyperref}
\usepackage{pgfplots}

\newtheorem{Proposition}{Proposition}
\newtheorem{Lemma}
{Lemma}
\newtheorem{Theorem}{Theorem}
\newtheorem{Remark}{Remark}

\usepackage{algorithm}
\usepackage{algpseudocode}
\usepackage[most]{tcolorbox}
\newtcbtheorem{mytcbprob}{Problem}{}{problem}
\usepackage{diagbox}
\usepackage{makecell}
\usepackage{amssymb}
\usepackage{bm}
\usepackage{multirow}
\allowdisplaybreaks

\makeatletter
\newenvironment{breakablealgorithm}
{
\begin{center} %
\refstepcounter{algorithm}
\hrule height.8pt depth0pt \kern2pt
\renewcommand{\caption}[2][\relax]{
{\raggedright\textbf{\ALG@name~\thealgorithm} ##2\par}

\ifx\relax##1\relax

\addcontentsline{loa}{algorithm}{\protect\numberline{\thealgorithm}##2}%

\else

\addcontentsline{loa}{algorithm}{\protect\numberline{\thealgorithm}##1}%

\fi

\kern2pt\hrule\kern2pt

}

}{

\kern2pt\hrule\relax

\end{center} %

}

\makeatother

\title{Distributed quantum algorithm for discrete logarithm problem}
\begin{document}
\begin{frontmatter}
\author{Hao Li $^{\rm a,b}$}
\author{Daowen Qiu $^{\rm a,b}$}
\address{$^{\rm a}$ Institute of Quantum Computing and Software, School of Computer Science and Engineering,  \\ Sun Yat-sen University, Guangzhou 510006, China;}
\address{$^{\rm b}$The Guangdong Key Laboratory of Information Security Technology, \\Sun Yat-sen University, Guangzhou 510006, China;}
\cortext[mycorrespondingauthor]{issqdw@mail.sysu.edu.cn (D. Qiu).}

\begin{abstract}
The  quantum algorithm with polynomial time for  discrete logarithm problem proposed by  Shor is  one of the most significant quantum algorithms,  but  a large number of qubits  may be required in the Noisy Intermediate-scale Quantum  (NISQ) era.
The main contributions of this paper are as follows: (1) A distributed  quantum algorithm for discrete logarithm problem is designed, and its space complexity and success probability exhibit certain advantages to some extent; (2) The classical error correction technique  proposed by Xiao and Qiu et al. is generalized by extending  three bits to  more than three bits.
\end{abstract}

\end{frontmatter}

\section{Introduction}

Quantum computing is developing rapidly,  showing advantages over classical computing in factoring larger integers \cite{shor1994algorithms}, searching unstructured database \cite{Grover1996}, solving linear system of equations \cite{harrow2009quantum}, and other fields. However, to implement  quantum algorithms in practice, medium or large scale general quantum computers are required. With the current state of the art in quantum physics, achieving such  quantum computers is still  difficult.  Therefore, in order to promote the application of quantum algorithms in the NISQ era, it is worthy of consideration for reducing the amount of qubits or other quantum resources required for quantum computers.
	
	Distributed quantum computing is a computational method that combines distributed computing with quantum computing  \cite{avron2021quantum,beals2013effcient,li2017application,yimsiriwattana2004distributed}. It aims to solve problems by utilizing multiple smaller quantum computers working in concert. Distributed quantum computing is usually used to reduce the resources required for each quantum computer, including the number of qubits, circuit depth and gate complexity. In light of these potential benefits, distributed quantum computing has been deeply studied  (e.g., \cite{avron2021quantum,beals2013effcient,li2017application,yimsiriwattana2004distributed,qiu2022distributed,tan2022distributed,Xiao2023DQAShor,Hao2023DGSP}).

We briefly review related contributions concerning distributed quantum computing. In 2004, the first distributed Shor’s algorithm was proposed by
Yimsiriwattana and Lomonaco et al. \cite{yimsiriwattana2004distributed}. In 2013, Beals and Brierley et al. proposed an algorithm for parallel addressing quantum memory \cite{beals2013effcient}.	In 2017,  Li and Qiu  et al.  proposed a distributed
phase estimation algorithm \cite{li2017application}.  In 2022, Qiu and Luo et al. proposed a distributed Grover's algorithm \cite{qiu2022distributed};  Tan  and Qiu et al. proposed a distributed quantum algorithm for Simon's problem \cite{tan2022distributed}.  In 2023,  Xiao  and Qiu et al.  proposed a new distributed Shor's algorithm \cite{Xiao2023DQAShor,Xiao2023DQAkShor}. In 2024,  Li  and Qiu et al. proposed  distributed exact quantum algorithms for  the Deutsch-Jozsa problem and the   generalized Simon’s problem  \cite{Hao2023DDJ,Hao2023DGSP}, as well as a distributed exact multi-objective quantum search algorithm \cite{Hao2024DEMA}.  In 2025, Li  and Qiu et al. proposed  a  distributed generalized Deutsch-Jozsa algorithm \cite{Li2025}. These distributed quantum algorithms can reduce quantum resources to some extent.
	
	The  quantum algorithm for discrete logarithm problem   is considered  to be one of the most significant algorithms in quantum computing \cite{shor1994algorithms},  with a time complexity of $O(L^3)$, a space complexity of $O(L)$, and a probability of success of $\dfrac{r-1}{r}(1-\epsilon)$,
 where $r$ is the least positive integer that satisfies $a^r\equiv 1 (\bmod\ N) $, $0<\epsilon<1$, and $L=\left\lfloor\log_2 N\right\rfloor+1$.  Since the best general classical deterministic algorithm for solving discrete logarithm problem has a time complexity of  $O(\sqrt{N}\log N)$ and a space complexity of $O(\sqrt{N})$ \cite{TR2022}, the  quantum algorithm for discrete logarithm problem demonstrates quantum advantages. 	
 
The  quantum algorithm for discrete logarithm problem can be applied to break  the Diffie-Hellman key exchange protocol \cite{DH1976}, which is an important protocol in public-key cryptosystems.   In 2020, for the problem of computing short discrete logarithms, Eker$\rm\mathring{a}$ proposed a different quantum algorithm   \cite{ME2020}. In 2022, for discrete logarithm problem  in a semigroup, Tinani and Rosenthal proposed a classical deterministic  algorithm \cite{TR2022}. 
	 In 2024, Hhan and Yamakawa et al. investigated the quantum computational complexity of  discrete logarithm problem \cite{Minki2024}.
		
	However, the quantum algorithm applied to solve discrete logarithm problem requires a large number of qubits. Therefore, it is timely to reduce the resources required  by designing novel approaches such as distributed quantum algorithms. 
	In this paper, our goal is to design a  distributed  quantum algorithm for  discrete logarithm problem, which has a  lower  space complexity than Shor's quantum algorithm.
	Furthermore, the success probability of our algorithm is higher than that of Shor's quantum algorithm. In our  algorithm, we generalize the classical error correction technique in the  algorithm proposed by Xiao and Qiu et al \cite{Xiao2023DQAkShor}, by extending  it from three bits to more than three bits. 
	
	The remaining part of the paper is organized as follows. In Section \ref{sec:preliminaries}, we review   the quantum Fourier transform, phase estimation algorithm,  and Shor's quantum algorithm for  discrete logarithm problem.  In Section \ref{sec:distributed discrete logarithmic quantum algorithm}, we  describe  a distributed  quantum algorithm for discrete logarithm problem, and analyze its correctness and complexity. Finally,  in Section \ref{sec:conclusions}, we conclude with a summary. 

\section{Preliminaries}\label{sec:preliminaries}
In this section, we  review the quantum Fourier transform, phase estimation algorithm, Shor's quantum algorithm for discrete logarithm problem, and other relevant concepts that will be used in the paper.

\subsection{Quantum Fourier transform}

	The quantum Fourier transform is a unitary operator that acts on the standard basis states as follows:
\begin{equation}
QFT |j\rangle=\frac{1}{\sqrt{2^n}}\sum_{k=0}^{2^n-1}e^{2\pi ijk/2^n}|k\rangle\text{,}
\end{equation}
for $j=0,1,\cdots,2^n-1$.

 The inverse quantum Fourier transform acts as follows:
\begin{equation}\label{inverse_QFT}
QFT^{-1} \frac{1}{\sqrt{2^n}}\sum_{k=0}^{2^n-1}e^{2\pi ijk/2^n}|k\rangle=|j\rangle\text{,}
\end{equation}
for $j=0,1,\cdots,2^n-1$.

The quantum Fourier transform and its inverse can be implemented using $O\left(n^2\right)$ elementary gates (i.e., $O\left(n^2\right)$ single-qubit and two-qubit gates) \cite{shor1994algorithms,nielsen2000quantum}.

\subsection{Phase estimation algorithm}

	Phase estimation algorithm is a practical application of the quantum Fourier transform.	
	 Consider a unitary operator $U$ and  a quantum state $|u\rangle$ such that 
\begin{equation}\label{eq:Uu}
U|u\rangle=e^{2\pi i\omega}|u\rangle\text{,}
\end{equation}
 for some real number $\omega\in[0,1)$.
 
  If we can implement controlled operation $C_m(U)$ satisfying that
\begin{equation}\label{CmU}
C_m(U)|j \rangle|u\rangle=|j\rangle U^j|u\rangle\text{,}
\end{equation}
for any positive integer $m$ and $m$-bit string $j$, 
 then we can apply the phase estimation algorithm to estimate $\omega$ (see Algorithm \ref{alg:PE}).  Fig. \ref{fig:controlled-U} shows the circuit diagram for the implementation of $C_m(U)$.
 
\begin{figure}[H]
	\centering
	\includegraphics[width=0.6\textwidth]{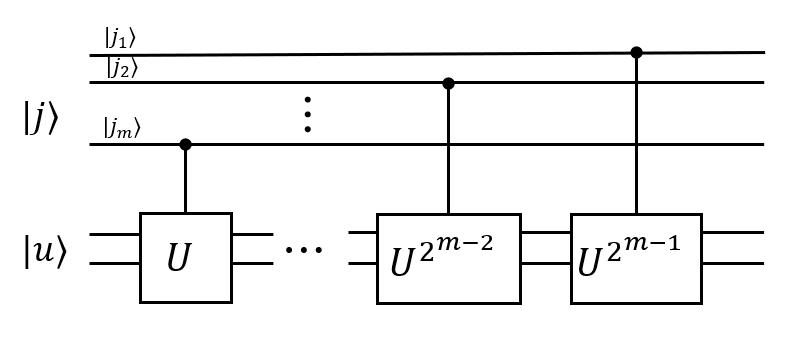}
	\caption{\label{fig:controlled-U}Implementation for $C_m(U)$.}
\end{figure}

For convenience, some notations are given below.
For any $x\in\{0,1\}^n$, denote
\begin{itemize}
\item $x_{[i,j]}=x_{i} x_{i+1}\cdots x_j$, where  $1\leq i\leq j \leq n$.
\item $|x|$: the length of string $x$.
\end{itemize}

For any real number $\omega=a_{1}a_{2}\cdots a_l.$ $b_1b_2\cdots$, denote
\begin{itemize}
\item $\omega_{[i,j]}=a_{i} a_{i+1}\cdots a_j$, where    $1\leq i\leq j \leq l$.
\item $\omega_{\{i,j\}}=b_i b_{i+1}\cdots b_j$, where $1\leq i\leq j$.
\item $|\omega|$: the absolute value of $\omega$.
\end{itemize}


 In this paper, for the sake of clarity and brevity, we take a bit string to be the same as its corresponding equivalent decimal number. In the following, we give a notation concerning the relationship between the distance between two bit strings.

For any $x,y\in\{0,1\}^t$, denote
\begin{equation}
 d_t(x,y)=\min\left(|x-y|, 2^t-|x-y|\right).
\end{equation}

In the following, we present an important lemma for analyzing the  algorithms in this paper. 

\begin{Lemma}[See \cite{Xiao2023DQAShor}]\label{d_t}
Let $x, y\in\{0,1\}^t$. Then, we have:  \\
{\rm (I)} $d_t(x,y)=\min_{b\in B}|b|$, where $B=\{b|(x+b)\bmod 2^{t}=y,  -(2^t-1)\leq b\leq 2^{t}-1\}$. \\
{\rm (II)} $d_t(\cdot,\cdot)$ is a distance on $\{0,1\}^t$.\\
{\rm (III)} If $d_t(x,y)<2^{t-t_0}$, then
$d_{t_0}\left(x_{[1,t_0]},y_{[1,t_0]}\right)\leq 1$,
where  $1\leq t_0<t$.
\end{Lemma}

In the following, we give a lemma related to the distance $d_t(\cdot,\cdot)$.

\begin{Lemma}\label{gd_t}
Let $x, y\in\{0,1\}^t$. 
 If $d_t(x,y)<2^{t-t_0}$, then
$d_{t_1}\left(x_{[1,t_1]},y_{[1,t_1]}\right)\leq 2^{t_1-t_0}$,
where  $1\leq t_0\leq t_1\leq t$.
\end{Lemma}

\begin{proof}
If $t_1=t$, the lemma clearly holds. If $t_1<t$, by (I) of  Lemma \ref{d_t} and $d_t(x,y)<2^{t-t_0}$, there exists an integer $b$ with $|b|<2^{t-t_0}$ such that
\begin{equation}
\left(2^{t-t_1} x_{[1,t_1]}+x_{[t_1+1,t]}+b\right)\bmod 2^t=2^{t-t_1} y_{[1,t_1]}+y_{[t_1+1,t]}.
\end{equation}
Then by (I) of  Lemma \ref{d_t} again, we have
\begin{align}
d_t\left(2^{t-t_1}x_{[1,t_1]},2^{t-t_1}y_{[1,t_1]}\right)&\leq \left|b+x_{[t_1+1,t]}-y_{[t_1+1,t]}\right|\\&\leq |b|+\left|x_{[t_1+1,t]}-y_{[t_1+1,t]}\right|\\
&<2^{t-t_0}+2^{t-t_1}.
\end{align}
Hence
\begin{equation}
d_{t_1}\left(x_{[1,t_1]},y_{[1,t_1]}\right)< 2^{t_1-t_0}+1,
\end{equation}
and the proposition holds.
\end{proof}

In the following, we describe the phase estimation algorithm that will be importantly used in this paper.

\vspace{0.5em}

\begin{breakablealgorithm}
\caption{Phase estimation algorithm}
\label{alg:PE}
\begin{algorithmic}[1]

\noindent\textbf{Input}: (1) An eigenstate $\ket{u}$ of the  unitary operator $U$ with eigenvalue $e^{2\pi i\omega}$, where $\omega\in[0,1)$.  (2) The controlled operator $C_t(U)$ such that $C_t(U)|j \rangle|u\rangle=|j\rangle U^j|u\rangle$, where $t=n+\left\lceil\log_2\left(2+\dfrac{1}{2\epsilon}\right)\right\rceil$,  $n\in\mathbb{N}^+$,  $\epsilon\in(0,1)$,  $j$ is a $t$-bit string.

\noindent\textbf{Output}: A $t$-bit string $\widetilde{\omega}$ such that $d_n\left(\widetilde{\omega}_{[1,n]},\omega_{\{1,n\}}\right)\leq 1$.

\noindent\textbf{Procedure}:

\State Create initial state $|0\rangle^{\otimes t}|u\rangle$.

\State Apply $H^{\otimes t}$ to the first register:

   $H^{\otimes t}|0\rangle^{\otimes t}|u\rangle=\dfrac{1}{\sqrt{2^t}}\sum\limits_{j=0}^{2^t-1}|j\rangle|u\rangle$.
\State Apply $C_t(U)$:

  $C_t(U)\dfrac{1}{\sqrt{2^t}}\sum\limits_{j=0}^{2^t-1}|j\rangle|u\rangle=\dfrac{1}{\sqrt{2^t}}\sum\limits_{j=0}^{2^t-1}|j\rangle U^j|u\rangle=\dfrac{1}{\sqrt{2^t}}\sum\limits_{j=0}^{2^t-1}|j\rangle e^{2\pi ij\omega}|u\rangle$.
  
\State Apply $QFT^{-1}$:

$QFT^{-1}\dfrac{1}{\sqrt{2^t}}\sum\limits_{j=0}^{2^t-1}e^{2\pi ij\omega}|j\rangle |u\rangle=\dfrac{1}{2^t}\sum\limits_{j=0}^{2^t-1}\sum\limits_{k=0}^{2^t-1}e^{2\pi ij(\omega-k/2^t)}|k\rangle |u\rangle$.
\State Measure the first register:

   Obtain a $t$-bit string $\widetilde{\omega}$.
\end{algorithmic}
\end{breakablealgorithm}

\vspace{0.5em}

	The goal of  the phase estimation algorithm is to estimate $\omega$, which can be more accurately described by the following propositions.

\begin{Proposition}[See \cite{nielsen2000quantum}]\label{phase_estimation_result}
	In  Algorithm \ref{alg:PE}, for any $n\in\mathbb{N}^+$ and any $\epsilon>0$, if $t=n+\left\lceil\log_2\left(2+\dfrac{1}{2\epsilon}\right)\right\rceil$, then the probability of $d_t\left(\widetilde{\omega},\omega_{\{1,t\}}\right)<2^{t-n}$ is at least $1-\epsilon$. 
\end{Proposition}

\begin{Proposition}\label{phase_estimation2}
	In  Algorithm \ref{alg:PE}, for any $n\in\mathbb{N}^+$ and any $\epsilon>0$, if $t=n+\left\lceil\log_2\left(2+\dfrac{1}{2\epsilon}\right)\right\rceil$, then the probability of $d_m\left(\widetilde{\omega}_{[1,m]},\omega_{\{1,m\}}\right)\leq 2^{m-n}$ is at least $1-\epsilon$, where  $n\leq m\leq t$.
\end{Proposition}
\begin{proof}
Immediate from Lemma \ref{gd_t} and Proposition \ref{phase_estimation_result}.
\end{proof}


\begin{Remark}
If the fractional part of $\omega$ does not exceed $t$ bits, by  Eq. (\ref{inverse_QFT}) and step 4 in Algorithm \ref{alg:PE}, then $\widetilde{\omega} $ is a perfect estimate of $\omega$. However, sometimes $\omega$ is not approximated by $\dfrac{\widetilde{\omega}}{2^t}$, but by $1-\dfrac{\widetilde{\omega}}{2^t}$. For example, if the binary representation of $\omega$ is $\omega=0.11\cdots1$ (sufficiently many 1s), we will  get the measurement $00\cdots 0$ with high probability, because at this point $e^{2\pi i \omega}$ is close to $e^{2\pi i 0}=1$.
\end{Remark}

\subsection{Quantum algorithm for discrete logarithm problem}
In this subsection, we review Shor's quantum algorithm for discrete logarithm problem \cite{nielsen2000quantum, kaye2006introduction}. The discrete logarithm problem is described below.

\begin{mytcbprob*}{discrete logarithm problem}
\textbf{Input:} A positive integer $N$, $\mathbb{Z}^*_N=\{x| 0\leq x< N, gcd(x,N)=1\}$, $a\in\mathbb{Z}^*_N$, $b\in\mathbb{Z}^*_N$, and the least positive integer $r$ that satisfies $a^r\equiv 1 (\bmod\ N) $.

\textbf{Promise:} $b\equiv a^g (\bmod\ N)$, $g\in\{0,1, \cdots,r-1\}$.

\textbf{Output:} The integer $g$.
\end{mytcbprob*}

For convenience and clarity in analyzing the algorithm, we assume that $r$ is prime and $r>2$. When $r=2$, since $g\in\{0,1\}$, it is only necessary to verify that  $b\equiv a^g (\bmod\ N)$ holds for  $g=0$ or $g=1$ in order to compute $g$. With some more sophisticated analysis, it can be verified that the  quantum algorithm for discrete logarithm problem is in fact also applicable to the case where $r$ is a composite number \cite{kaye2006introduction}. 

The important unitary operators  $M_a$  and $M_b$  in Shor's quantum algorithm for discrete logarithm problem are defined as 
\begin{align}
M_a|x \rangle=|ax\ (\bmod\ N)\rangle \text{,}\\
M_b|x \rangle=|bx\ (\bmod\ N)\rangle \text{.}
\end{align}

Denote
\begin{equation}
|u_s\rangle=\dfrac{1}{\sqrt{r}}\sum\limits_{k=0}^{r-1}e^{-2\pi i\frac{s}{r}k}|a^k (\bmod\ N)\rangle,
\end{equation}
where $s\in\{0,1, \cdots,r-1\}$.

We have
\begin{align}
M_a|u_s\rangle=e^{2\pi i\frac{s}{r}} |u_s\rangle,\label{M_a}\\
M_b|u_s\rangle=e^{2\pi i\frac{sg}{r}} |u_s\rangle,\label{M_b}\\
\dfrac{1}{\sqrt{r}}\sum\limits_{s=0}^{r-1}|u_s\rangle=|1\rangle,
\end{align} and
\begin{equation}\label{us_orthonormal}
\langle u_s|u_{s'}\rangle=
\begin{cases} 0 &\text{if $s\not= s'$},\\
1 &\text{if $s=s'$}.
\end{cases}
\end{equation}

In addition, we define the controlled  operator $C^*_m(U)$ as follows:
\begin{equation}
C^*_t(U)|j \rangle|k \rangle|u\rangle=|j\rangle|k\rangle U^j|u\rangle\text{,}
\end{equation}
where $t\in \mathbb{N}^+$, $j$ and $k$ are both $t$-bit strings, $U$ is a unitary operator, and $\ket{u}$ is a quantum state. 

Algorithm \ref{alg:DLQA}   and Fig. \ref{fig:discrete logarithmic quantum algorithm} show the procedure of Shor's  quantum algorithm  \cite{nielsen2000quantum}.

\vspace{1em}

\begin{breakablealgorithm}
\caption{Quantum algorithm for  discrete logarithm problem}
\label{alg:DLQA}
\begin{algorithmic}[1]
\noindent\textbf{Input}: A positive integer $N$, $\mathbb{Z}^*_N=\{x| 0\leq x< N, gcd(x,N)=1\}$, $a\in\mathbb{Z}^*_N$, $b\in\mathbb{Z}^*_N$, and  the least positive integer $r$ that satisfies $a^r\equiv 1 (\bmod\ N) $, where $b\equiv  a^g (\bmod\ N)$, $g\in\{0,1,\cdots,r-1\}$.

\noindent\textbf{Output}: The integer $\hat{g}$ that satisfies $b\equiv  a^{\hat{g}} (\bmod\ N)$.

\noindent\textbf{Procedure}:
\State Create initial state $|0\rangle^{\otimes t}|0\rangle^{\otimes t}|1\rangle$:

$t=\lceil\log_2r+1\rceil+\left\lceil\log_2\left(2+\dfrac{1}{\epsilon}\right)\right\rceil$, the third register is  $L$-qubit, and $L=\left\lfloor\log_2 N\right\rfloor+1$.
\State Apply $H^{\otimes t}$ to the first  and the second registers:

$\left(H^{\otimes t}\otimes H^{\otimes t}\otimes I^{\otimes L}\right)|0\rangle^{\otimes t}|0\rangle^{\otimes t}|1\rangle=  \dfrac{1}{\sqrt{2^t}}\sum\limits_{j=0}^{2^t-1}|j\rangle\dfrac{1}{\sqrt{2^t}}\sum\limits_{k=0}^{2^t-1}|k\rangle|1\rangle$.

\State Apply $C^*_t(M_a)$ to the first  and the third registers, and $C_t(M_b)$ to the second and third registers:
\begin{align*}
& \left(\left(I^{\otimes t}\otimes C_{t}\left(M_b\right)\right)C^*_{t}\left(M_a\right)\right)\dfrac{1}{\sqrt{2^t}}\sum\limits_{j=0}^{2^t-1}|j\rangle\dfrac{1}{\sqrt{2^t}}\sum\limits_{k=0}^{2^t-1}|k\rangle|1\rangle \\=&\left(\left(I^{\otimes t}\otimes C_{t}\left(M_b\right)\right)C^*_{t}\left(M_a\right)\right)\dfrac{1}{\sqrt{2^t}}\sum\limits_{j=0}^{2^t-1}|j\rangle\dfrac{1}{\sqrt{2^t}}\sum\limits_{k=0}^{2^t-1}|k\rangle\left(\dfrac{1}{\sqrt{r}}\sum\limits_{s=0}^{r-1}|u_s\rangle\right)
 \\=&\dfrac{1}{\sqrt{r}}\sum\limits_{s=0}^{r-1} \dfrac{1}{\sqrt{2^{t}}}\sum\limits_{j=0}^{2^{t}-1}e^{2\pi ij\frac{s}{r}}|j\rangle  \dfrac{1}{\sqrt{2^{t}}}\sum\limits_{k=0}^{2^{t}-1}e^{2\pi ik\frac{sg}{r}}|k\rangle  |u_s\rangle.
\end{align*}
 
\State Apply $QFT^{-1}$ to the first  and the second registers:
\begin{align*}
&\left(QFT^{-1}\otimes QFT^{-1}\otimes I^{\otimes L}\right)\dfrac{1}{\sqrt{r}}\sum\limits_{s=0}^{r-1} \dfrac{1}{\sqrt{2^{t}}}\sum\limits_{j=0}^{2^{t}-1}e^{2\pi ij\frac{s}{r}}|j\rangle  \dfrac{1}{\sqrt{2^{t}}}\sum\limits_{k=0}^{2^{t}-1}e^{2\pi ik\frac{sg}{r}}|k\rangle  |u_s\rangle \\=&  \dfrac{1}{\sqrt{r}}\sum\limits_{s=0}^{r-1} \dfrac{1}{2^{t}}\sum\limits_{j=0}^{2^{t}-1}\sum\limits_{j'=0}^{2^{t}-1}e^{2\pi ij\left(\frac{s}{r}-\frac{j'}{2^t}\right)}|j'\rangle \dfrac{1}{2^{t}}\sum\limits_{k=0}^{2^{t}-1}\sum\limits_{k'=0}^{2^{t}-1}e^{2\pi ik\left(\frac{sg}{r}-\frac{k'}{2^t}\right)}|k'\rangle |u_s\rangle.
 \end{align*}
 
\State Measure the first and the second register:

obtain a $t$-bit string $m_a$ and  a $t$-bit string $m_b$.
\State Round $m_ar/2^t$ and $m_br/2^t$ to get integers $\hat{m}_a$ and $\hat{m}_b$ respectively. If $\hat{m}_a=0$, then go to step 1; otherwise, compute $\hat{g}=\hat{m}_a^{-1}\hat{m}_b (\bmod\ r)$. If
$b\equiv a^{\hat{g}} (\bmod\ N)$, then output $\hat{g}$, otherwise  go to step 1.
\end{algorithmic}
\end{breakablealgorithm}

\begin{figure}[H]
	\centering
	\includegraphics[width=0.5\textwidth]{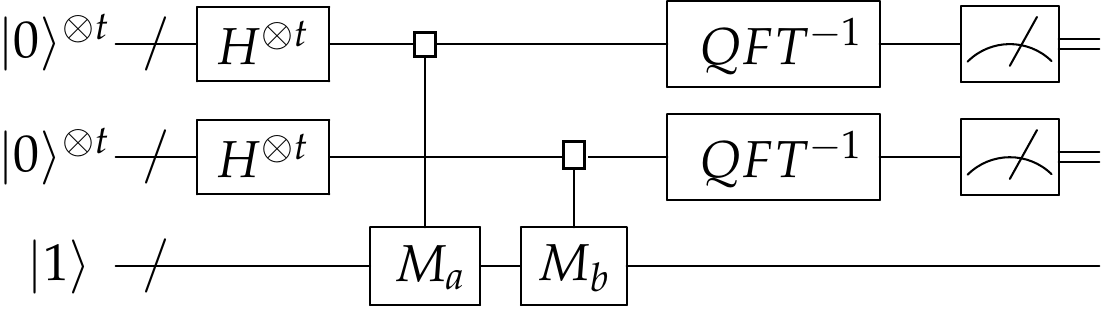}
	\caption{\label{fig:discrete logarithmic quantum algorithm} Circuit for Shor's quantum algorithm.}
\end{figure}

	The function of the quantum part of Algorithm \ref{alg:DLQA} (step 1 to 5) and the classical part of Algorithm \ref{alg:DLQA} (step 6) can be described by the following proposition.
	
\begin{Proposition}[See \cite{nielsen2000quantum,kaye2006introduction}]\label{DLQAResult}
In step 5 of Algorithm \ref{alg:DLQA},  the probability that $m_a$ and $m_b$ satisfy
\begin{equation*}
\left(\left|\dfrac{m_a}{2^{t}}-\dfrac{s}{r}\right|\leq2^{-\lceil\log_2r+1\rceil}\right)\bigcap \left(\left|\dfrac{m_b}{2^{t}}-\dfrac{sg(\bmod\ r)}{r}\right|\leq2^{-\lceil\log_2r+1\rceil}\right)
\end{equation*}
 is at least $1-\epsilon$, where $s\in \{0,1,\cdots,r-1\}$.

In step 6 of Algorithm \ref{alg:DLQA}, the probability that $b\equiv a^{\hat{g}} (\bmod\ N)$ holds is at least $\dfrac{r-1}{r}(1-\epsilon)$.
\end{Proposition}

\section{Distributed quantum algorithm for  discrete logarithm problem} \label{sec:distributed discrete logarithmic quantum algorithm}
		
	In 2017, Li and Qiu et al. \cite{li2017application} proposed a distributed phase estimation algorithm. In 2023, Xiao and Qiu et al. \cite{Xiao2023DQAShor} proposed a distributed Shor's algorithm with two compute nodes by combining the classical error correction technique. Furthermore, Xiao and Qiu et al. proposed a novel  distributed phase estimation algorithm and distributed Shor's algorithm with multiple compute nodes \cite{Xiao2023DQAkShor}.	
	
	Drawing on the ideas and methods of the distributed phase estimation algorithm and distributed Shor's algorithm  proposed by Xiao and Qiu et al. \cite{Xiao2023DQAkShor}, we  design a distributed quantum algorithm for   discrete logarithm problem. 
	
\subsection{Algorithm design}
	
In the following, we describe the  designing idea of our   distributed quantum algorithm.

Let $l_1, l_2, \cdots,l_{k+1}$ satisfy 
\begin{equation}
1=l_1<l_2<\cdots <l_{k+1}=\lceil\log_2r+1\rceil+1, 
\end{equation}
where 
 $l_i=\left\lfloor (i-1)\cdot\dfrac{\lceil\log_2r+1\rceil+1}{k}\right\rfloor$ and  $i=2,\cdots,k$.
 
We use $k$ compute nodes (denoted as $T_1, T_2, \cdots,T_k$) to estimate the bits of different parts of $s/r$ and $sg(\bmod\ r)/r$ respectively, where  $T_i$ estimates ${(s/r)}_{\{l_i,l_{i+1}+h\}}$ and ${(sg(\bmod\ r)/r)}_{\{l_i,l_{i+1}+h\}}$, where $s\in\{0,1,\cdots, r\}$, $i=1,2,\cdots,k-1$ and $2\leq h\leq \left\lfloor \dfrac{\lceil\log_2r+1\rceil+1}{k}\right\rfloor$. $T_{k+1}$ estimates ${(s/r)}_{\{l_k,l_{k+1}\}}$ and ${(sg(\bmod\ r)/r)}_{\{l_k,l_{k+1}\}}$. (shown in Fig. \ref{DDLPAestimation}). 

\begin{figure}[H]
	\centering
	\includegraphics[width=0.9\textwidth]{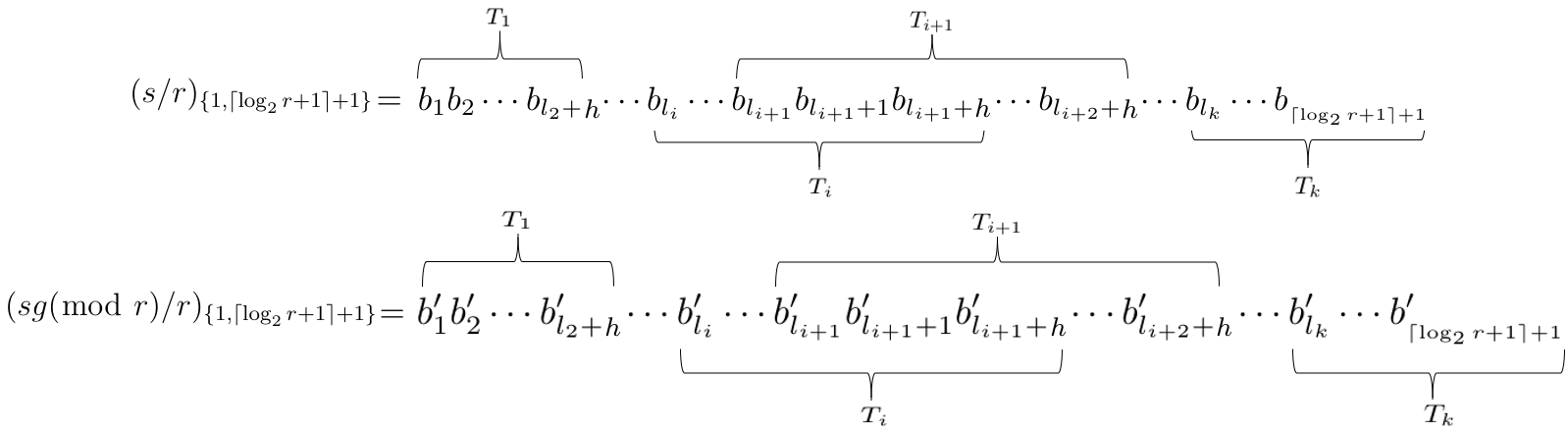}
	\caption{\label{DDLPAestimation} The positions of bits estimated by each compute node (Algorithm \ref{alg:DDLQA}).}
\end{figure}

We utilize the following technique to estimate the bits of different parts of $s/r$ and $sg(\bmod\ r)/r$. For any real number $\omega$, suppose its binary representation is $\omega=a_{1}a_{2}\cdots a_l.b_1b_2\cdots$, and denote 
\begin{equation}
\omega_{\{i,+\infty\}}=0.b_ib_{i+1}\cdots, 
\end{equation}
where $i\in \mathbb{N}^+$.

Based on $U|u\rangle=e^{2\pi i\omega}|u\rangle$. Note that 
\begin{align}
U^{2^{n-1}}|u\rangle=&e^{2\pi i\left(2^{n-1}\omega\right)}|u\rangle
=e^{2\pi i\omega_{\{n,+\infty\}}}|u\rangle, 
\end{align}
where $n\in\mathbb{N}^+$.

Thus, if estimating $\omega_{\{n,+\infty\}}$, we can apply the phase estimation algorithm similarly and change $C_t(U)$ in step 3 of Algorithm \ref{alg:PE} to $C_t\left(U^{2^{n-1}}\right)$ accordingly \cite{li2017application}.

The corresponding positions of the first $(h+1)$-bits of $T_{i+1}$'s estimation and the last $(h+1)$-bits of $T_{i}$'s estimation are overlapping with each other. Thus, we can use the bits with overlapped positions to correct the estimation results and eventually concatenate all the estimates.

Before presenting our  algorithm, we introduce the subroutine used to perform error correction and concatenation of the measurements, i.e., Algorithm \ref{alg:CorrectAndCombine}.  
For any bit string $x$ and integer $b$,  we denote  $Sum(x,b)$  as follows:
\begin{itemize}
\item  $Sum(x,b)$ is a bit string $y$ of length $|x|$, where $\left(x+b\right)\bmod 2^{|x|}=y$.
\end{itemize}

\vspace{0.5em}

\begin{breakablealgorithm}
\caption{Correct}
\label{alg:CorrectAndCombine}
\begin{algorithmic}[1]
\noindent\textbf{Input}: $k$ bit strings $m_1,\cdots,m_k$.

\noindent\textbf{Output}: The bit string $c_1$.

\noindent\textbf{Procedure}:
\State Set $c_k=m_k$.
\For{$j=k-1$ to $1$}
	\State Choose $q_j\in\{0, \pm 1,\pm 2,\cdots,\pm 2^{h-1}\}$ such that\\
	\quad\quad  $Sum\left(\left(m_j\right)_{\left[l_{j+1},l_{j+1}+h\right]},q_j\right)=$$\left(c_{j+1}\right)_{[1,h+1]}$.
	\State $p_j\leftarrow Sum\left(m_j,q_j\right)$.
	\State $c_j\leftarrow p_j\circ \left(c_{j+1}\right)_{\left[h+2,\left|c_{j+1}\right|\right]}$ (``$\circ$" represents catenation).
\EndFor
\end{algorithmic}
\end{breakablealgorithm}

\vspace{0.5em}

The analysis of the correctness and effects of Algorithm \ref{alg:CorrectAndCombine} are given in Proposition \ref{prop:CorrectAndCombine} in Section \ref{DDLQA Algorithm correctness analysis}.

The distributed quantum algorithm is shown in Algorithm \ref{alg:DDLQA}. Fig. \ref{fig:distributed discrete logarithmic quantum algorithm} shows the quantum circuit of Algorithm \ref{alg:DDLQA}. 
To describe  Algorithm \ref{alg:DDLQA} succinctly, we denote 
\begin{itemize}
\item $|\psi_{t,\omega}\rangle=\dfrac{1}{2^t}\sum\limits_{u=0}^{2^t-1}\sum\limits_{v=0}^{2^t-1}e^{2\pi iu(\omega-v/2^t)}|v\rangle$, where $t\in \mathbb{N}^+$ and $\omega \in[0,1)$.
\end{itemize}

\vspace{2em}

\begin{breakablealgorithm}
\caption{Distributed quantum algorithm for  discrete logarithm problem}
\label{alg:DDLQA}
\begin{algorithmic}[1]
\noindent\textbf{Input}: A positive integer $N$, $\mathbb{Z}^*_N=\{x| 0\leq x< N, gcd(x,N)=1\}$, $a\in\mathbb{Z}^*_N$, $b\in\mathbb{Z}^*_N$, and  the least positive integer $r$ that satisfies $a^r\equiv 1 (\bmod\ N) $, where $b\equiv  a^g (\bmod\ N)$, $g\in\{0,1,\cdots,r-1\}$.

\noindent\textbf{Output}: The integer $\hat{g}$ that satisfies $b\equiv  a^{\hat{g}} (\bmod\ N)$.

\noindent\textbf{Procedure}:
\State Node $T_1,\cdots,T_k$ creates initial state $|0\rangle^{\otimes t_1}_{1a}|0\rangle^{\otimes t_1}_{1b}|1\rangle_C$, $\cdots,|0\rangle^{\otimes t_k}_{ka}|0\rangle^{\otimes t_k}_{kb}$ respectively, where  register ${ja}$ and register ${jb}$ are both $t_j$-qubit, $t_j=l_{j+1}+3-l_j+\left\lceil\log_2\left(2+\dfrac{k}{\epsilon'}\right)\right\rceil$, $j=1,\cdots,k-1$,  $t_k=l_{k+1}+1-l_k+\left\lceil\log_2\left(2+\dfrac{k}{\epsilon'}\right)\right\rceil$, $0<\epsilon'<\epsilon<1$, register $C$ is  $L$-qubit, and $L=\left\lfloor\log_2 N\right\rfloor+1$.

Denote 
\begin{align}
\ket{\phi_1}=|0\rangle^{\otimes t_1}_{1a}|0\rangle^{\otimes t_1}_{1b}|1\rangle_C\cdots|0\rangle^{\otimes t_k}_{ka}|0\rangle^{\otimes t_k}_{kb}.\nonumber
\end{align}

\State Set $u=1$.

\noindent\textbf{Node  $T_u$ excutes}:
\State Apply $H^{\otimes t_u}$ to registers ${ua}$ and ${ub}$ of $\ket{\phi_u}$:
\begin{align}
\ket{\phi_{u+1}}=\left(I^{\otimes 2\sum\nolimits_{i=1}^{u-1}t_i}\otimes \left(H^{\otimes t_u}\otimes H^{\otimes t_u}\otimes I^{\otimes L}\right) \otimes  I^{\otimes 2\sum\nolimits_{i=u+1}^{k}t_i}\right)\ket{\phi_u}.\nonumber
\end{align}

\State Apply $C^*_{t_u}\left(M_a^{2^{l_u-1}}\right)$ to registers ${ua}$ and $C$ of  $\ket{\phi_{u+1}}$, and $C_{t_u}(M_b^{2^{l_u-1}})$ to registers ${ub}$ and $C$ of  $\ket{\phi_{u+1}}$:
\begin{align}
\ket{\phi_{u+2}}=\left(I^{\otimes 2\sum\nolimits_{i=1}^{u-1}t_i}\otimes \left(\left(I^{\otimes t_u}\otimes C_{t_u}\left(M_b^{2^{l_u-1}}\right)\right)C^*_{t_u}\left(M_a^{2^{l_u-1}}\right)\right) \otimes I^{\otimes 2\sum\nolimits_{i=u+1}^{k}t_i}\right)\ket{\phi_{u+1}}.\nonumber
\end{align}

\State  Apply $QFT^{-1}$ to registers ${ua}$ and ${ub}$ of  $\ket{\phi_{u+2}}$:
\begin{align}
\ket{\phi_{u+3}}=\left(I^{\otimes 2\sum\nolimits_{i=1}^{u-1}t_i}\otimes \left(QFT^{-1}\otimes QFT^{-1}\otimes I^{\otimes L}\right) \otimes  I^{\otimes 2\sum\nolimits_{i=u+1}^{k}t_i}\right)\ket{\phi_{u+2}}.\nonumber
\end{align}

\State If $u<k$, then teleport the qubits of register $C$ to node $T_{u+1}$. Denote this quantum state as  $\ket{\phi_{u+4}}$. Set $u\leftarrow u+1$ and go to step 3.

\noindent\textbf{Finally}:

\State Node $T_j$ measures the first $l_{j+1}+h+1-l_j$ qubits of its register ${ja}$ and register ${jb}$, and denotes the measurement results as $m_{ja}$ and $m_{jb}$, respectively, where $2\leq h\leq \left\lfloor \dfrac{\lceil\log_2r+1\rceil+1}{k}\right\rfloor$, $j=1,\cdots,k-1$. Node $T_k$ measures the first $l_{k+1}+1-l_k$ qubits of its register ${ka}$ and register ${kb}$, and denotes the measurement results as $m_{ka}$ and $m_{kb}$, respectively.

\State $m_a\leftarrow Correct(m_{1a},\cdots,m_{ka})$: $m_a$ is an $\lceil\log_2r+1\rceil$+1-bit string.

\noindent$m_b\leftarrow Correct(m_{1b},\cdots,m_{kb})$: $m_b$ is an $\lceil\log_2r+1\rceil$+1-bit string.

\State Round $m_ar/2^{|m_a|}$ and $m_br/2^{|m_b|}$ to get integers $\hat{m}_a$ and $\hat{m}_b$ respectively. If $\hat{m}_a=0$, then go to step 1; otherwise, compute $\hat{g}=\hat{m}_a^{-1}\hat{m}_b (\bmod\ r)$. If
$b\equiv a^{\hat{g}} (\bmod\ N)$, then output $\hat{g}$, otherwise  go to step 1.
\end{algorithmic}
\end{breakablealgorithm}

\begin{figure}[h]
	\centering
	\includegraphics[width=\textwidth]{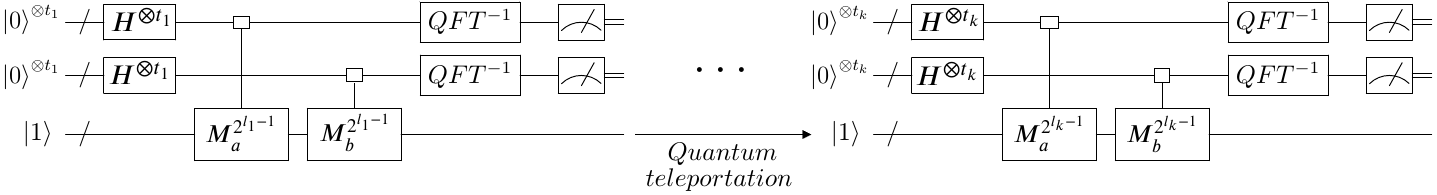}
	\caption{\label{fig:distributed discrete logarithmic quantum algorithm} Circuit for  distributed quantum algorithm.}
\end{figure}

\subsection{Algorithm  correctness analysis}\label{DDLQA Algorithm correctness analysis}
In the following, we analyze the correctness of Algorithm \ref{alg:DDLQA}. We first give a lemma related to proving the correctness of  Algorithm \ref{alg:DDLQA}.

\begin{Lemma}\label{psifinal}
The quantum state at step 7 of Algorithm \ref{alg:DDLQA} is
\begin{equation}
\ket{\Phi}=\dfrac{1}{\sqrt{r}}\sum\limits_{s=0}^{r-1}\left(\mathop{\otimes}\limits_{j=1}^{k}\Ket{\psi_{t_j,\left(\frac{s}{r}\right)_{\left\{l_j,+\infty\right\}}}}_{ja}\Ket{\psi_{t_j,\left(\frac{sg}{r}\right)_{\left\{l_j,+\infty\right\}}}}_{{jb}}\right) |u_{s}\rangle_C.
\end{equation}
\end{Lemma}
\begin{proof}
First, in step 3 of Algorithm \ref{alg:DDLQA}, we have
\begin{align}
&\ket{\phi_{u+1}}\\=&\left(I^{\otimes 2\sum\nolimits_{i=1}^{u-1}t_i}\otimes (H^{\otimes t_u}\otimes H^{\otimes t_u}\otimes I^{\otimes L}) \otimes  I^{\otimes 2\sum\nolimits_{i=u+1}^{k}t_i}\right)\ket{\phi_u}
\\=&\dfrac{1}{\sqrt{r}}\sum\limits_{s=0}^{r-1}\left(\mathop{\otimes}\limits_{j=1}^{u-1}\Ket{\psi_{t_j,\left(\frac{s}{r}\right)_{\left\{l_j,+\infty\right\}}}}_{{ja}}\Ket{\psi_{t_j,\left(\frac{sg}{r}\right)_{\left\{l_j,+\infty\right\}}}}_{{jb}}\right) H^{\otimes t_u}|0\rangle^{\otimes t_u}_{{ua}}H^{\otimes t_u}|0\rangle^{\otimes t_u}_{{ub}}|u_s\rangle_C\left(\mathop{\otimes}\limits_{j=u+1}^k|0\rangle^{\otimes t_j}_{{ja}}|0\rangle^{\otimes t_j}_{{jb}}\right)\\
=&\dfrac{1}{\sqrt{r}}\sum\limits_{s=0}^{r-1}\left(\mathop{\otimes}\limits_{j=1}^{u-1}\Ket{\psi_{t_j,\left(\frac{s}{r}\right)_{\left\{l_j,+\infty\right\}}}}_{{ja}}\Ket{\psi_{t_j,\left(\frac{sg}{r}\right)_{\left\{l_j,+\infty\right\}}}}_{{jb}}\right) \left(\dfrac{1}{\sqrt{2^{t_u}}}\sum\limits_{j=0}^{2^{t_u}-1}|j\rangle_{{ua}}\right)\left(\dfrac{1}{\sqrt{2^{t_u}}}\sum\limits_{j=0}^{2^{t_u}-1}|j\rangle_{{ub}}\right)|u_s\rangle_C\nonumber\\&\left(\mathop{\otimes}\limits_{j=u+1}^k|0\rangle^{\otimes t_j}_{{ja}}|0\rangle^{\otimes t_j}_{{jb}}\right).
\end{align}

Then, in step 4 of Algorithm \ref{alg:DDLQA}, we can get
\begin{align}
\ket{\phi_{u+2}}=&\left(I^{\otimes 2\sum\nolimits_{i=1}^{u-1}t_i}\otimes \left(\left(I^{\otimes t_u}\otimes C_{t_u}\left(M_b^{2^{l_u-1}}\right)\right)C^{*}_{t_u}\left(M_a^{2^{l_u-1}}\right)\right) \otimes I^{\otimes 2\sum\nolimits_{i=u+1}^{k}t_i}\right)\ket{\phi_{u+1}}\\
=&\dfrac{1}{\sqrt{r}}\sum\limits_{s=0}^{r-1}\left(\mathop{\otimes}\limits_{j=1}^{u-1}\Ket{\psi_{t_j,\left(\frac{s}{r}\right)_{\left\{l_j,+\infty\right\}}}}_{{ja}}\Ket{\psi_{t_j,\left(\frac{sg}{r}\right)_{\left\{l_j,+\infty\right\}}}}_{{jb}}\right) \left(\dfrac{1}{\sqrt{2^{t_u}}}\sum\limits_{j=0}^{2^{t_u}-1}e^{2\pi ij\left(\frac{s}{r}\right)_{\{l_u,+\infty\}}}|j\rangle_{{ua}}\right)\nonumber\\&\left(\dfrac{1}{\sqrt{2^{t_u}}}\sum\limits_{j=0}^{2^{t_u}-1}e^{2\pi ij\left(\frac{sg}{r}\right)_{\{l_u,+\infty\}}}|j\rangle_{{ub}}\right)|u_s\rangle_C\left(\mathop{\otimes}\limits_{j=u+1}^k|0\rangle^{\otimes t_j}_{{ja}}|0\rangle^{\otimes t_j}_{{jb}}\right).
\end{align}

After that, in step 5 of Algorithm \ref{alg:DDLQA}, the following state is obtained
\begin{align}
\ket{\phi_{u+3}}=&\left(I^{\otimes 2\sum\nolimits_{i=1}^{u-1}t_i}\otimes \left(QFT^{-1}\otimes QFT^{-1}\otimes I^{\otimes L}\right) \otimes  I^{\otimes 2\sum\nolimits_{i=u+1}^{k}t_i}\right)\ket{\phi_{u+2}}\\
=&\dfrac{1}{\sqrt{r}}\sum\limits_{s=0}^{r-1}\left(\mathop{\otimes}\limits_{j=1}^{u}\Ket{\psi_{t_j,\left(\frac{s}{r}\right)_{\left\{l_j,+\infty\right\}}}}_{{ja}}\Ket{\psi_{t_j,\left(\frac{sg}{r}\right)_{\left\{l_j,+\infty\right\}}}}_{{jb}}\right) |u_s\rangle_C\left(\mathop{\otimes}\limits_{j=u+1}^k|0\rangle^{\otimes t_j}_{{ja}}|0\rangle^{\otimes t_j}_{{jb}}\right).
\end{align}

Furthermore, in step 6 of Algorithm \ref{alg:DDLQA}, it leads to the
following state
\begin{align}
\ket{\phi_{u+4}}=&\dfrac{1}{\sqrt{r}}\sum\limits_{s=0}^{r-1}\left(\mathop{\otimes}\limits_{j=1}^{u}\Ket{\psi_{t_j,\left(\frac{s}{r}\right)_{\left\{l_j,+\infty\right\}}}}_{{ja}}\Ket{\psi_{t_j,\left(\frac{sg}{r}\right)_{\left\{l_j,+\infty\right\}}}}_{{jb}}\right)|0\rangle^{\otimes t_{u+1}}_{{ja}}|0\rangle^{\otimes t_{u+1}}_{{jb}}|u_s\rangle_C\left(\mathop{\otimes}\limits_{j=u+2}^k|0\rangle^{\otimes t_j}_{{ja}}|0\rangle^{\otimes t_j}_{{jb}}\right).
\end{align}

Thus, when $u=k$, we can obtain that the quantum state at step 7 of Algorithm \ref{alg:DDLQA} is
\begin{equation}
\ket{\Phi}=\dfrac{1}{\sqrt{r}}\sum\limits_{s=0}^{r-1}\left(\mathop{\otimes}\limits_{j=1}^{k}\Ket{\psi_{t_j,\left(\frac{s}{r}\right)_{\left\{l_j,+\infty\right\}}}}_{ja}\Ket{\psi_{t_j,\left(\frac{sg}{r}\right)_{\left\{l_j,+\infty\right\}}}}_{{jb}}\right) |u_{s}\rangle_C.
\end{equation}
\end{proof}

In the following, we  give another lemma related to proving the correctness of  Algorithm \ref{alg:DDLQA}.

\begin{Lemma}\label{distancestring}
In step 7 of Algorithm \ref{alg:DDLQA},  the probability that the measurements $m_{ja}$ and $m_{jb}$  satisfy
\begin{equation}
\begin{split}
&\left(d_{|m_{ja}|}\left(m_{ja},\left(\frac{s}{r}\right)_{\left\{l_j,l_{j+1}+h\right\}}\right)\leq 2^{h-2}\right)\bigcap\left(d_{|m_{jb}|}\left(m_{jb},\left(\frac{sg}{r}\right)_{\left\{l_j,l_{j+1}+h\right\}}\right)\leq 2^{h-2}\right)\\&\bigcap
\left(d_{|m_{ka}|}\left(m_{ka},\left(\frac{s}{r}\right)_{\left\{l_k,l_{k+1}\right\}}\right)\leq 1\right)\bigcap\left(d_{|m_{kb}|}\left(m_{kb},\left(\frac{sg}{r}\right)_{\left\{l_k,l_{k+1}\right\}}\right)\leq 1\right)
\end{split}
\end{equation}
for all $j\in\{1,\cdots,k-1\}$ is greater than $1-\epsilon$, where $s\in \{0,1,\cdots,r-1\}$.
\end{Lemma}
\begin{proof}
By Lemma \ref{psifinal}, the quantum state at step 7 of Algorithm \ref{alg:DDLQA} is 
\begin{equation}
\ket{\Phi}=\dfrac{1}{\sqrt{r}}\sum\limits_{s=0}^{r-1}\left(\mathop{\otimes}\limits_{j=1}^{k}\Ket{\psi_{t_j,\left(\frac{s}{r}\right)_{\left\{l_j,+\infty\right\}}}}_{ja}\Ket{\psi_{t_j,\left(\frac{sg}{r}\right)_{\left\{l_j,+\infty\right\}}}}_{{jb}}\right) |u_{s}\rangle_C.
\end{equation}

Denote 
\begin{equation}
A_{s,j}=\left\{x\in\{0,1\}^{t_j}\Bigg| d_{|m_{ja}|}\left(x_{[1,|m_{ja}|]},\left(\frac{s}{r}\right)_{\{l_j,l_{j+1}+h\}}\right)\leq 2^{h-2}\right\}. 
\end{equation}

Let 
\begin{equation}
Q_{A_{s,j}}=\sum_{i\in A_{s,j}}|i\rangle\langle i|,
\end{equation}
and
\begin{equation}
Q_{A_j}=\sum_{i\in \mathop{\bigcup}\limits_{s=0}^{r-1} A_{s,j}}|i\rangle\langle i|. 
\end{equation}

Denote 
\begin{equation}
A_{s,k}=\left\{x\in\{0,1\}^{t_k}\Bigg| d_{|m_{ka}|}\left(x_{[1,|m_{ka}|]},\left(\frac{s}{r}\right)_{\{l_k,l_{k+1}\}}\right)\leq 1\right\}. 
\end{equation}

Let 
\begin{equation}
Q_{A_{s,k}}=\sum_{i\in A_{s,k}}|i\rangle\langle i|,
\end{equation}
and
\begin{equation}
Q_{A_k}=\sum_{i\in \mathop{\bigcup}\limits_{s=0}^{r-1} A_{s,k}}|i\rangle\langle i|. 
\end{equation}

Denote 
\begin{equation}
B_{s,j}=\left\{x\in\{0,1\}^{t_j}\Bigg|d_{|m_{jb}|}\left(x_{[1,|m_{jb}|]},\left(\frac{sg}{r}\right)_{\{l_j,l_{j+1}+h\}}\right)\leq 2^{h-2}\right\}. 
\end{equation}

Let 
\begin{equation}
Q_{B_{s,j}}=\sum_{i\in  B_{s,j}}|i\rangle\langle i|,
\end{equation}
and
\begin{equation}
Q_{B_j}=\sum_{i\in \mathop{\bigcup}\limits_{s=0}^{r-1} B_{s,j}}|i\rangle\langle i|. 
\end{equation}

Denote 
\begin{equation}
B_{s,k}=\left\{x\in\{0,1\}^{t_k}\Bigg|d_{|m_{kb}|}\left(x_{[1,|m_{kb}|]},\left(\frac{sg}{r}\right)_{\{l_k,l_{k+1}\}}\right)\leq 1\right\}. 
\end{equation}

Let 
\begin{equation}
Q_{B_{s,k}}=\sum_{i\in  B_{s,k}}|i\rangle\langle i|,
\end{equation}
and
\begin{equation}
Q_{B_k}=\sum_{i\in \mathop{\bigcup}\limits_{s=0}^{r-1} B_{s,k}}|i\rangle\langle i|. 
\end{equation}

In step 7 of Algorithm \ref{alg:DDLQA},  the probability that the measurements $m_{ja}$ and $m_{jb}$  satisfy
\begin{equation}
 \begin{split}
&\left(d_{|m_{ja}|}\left(m_{ja},\left(\frac{s}{r}\right)_{\left\{l_j,l_{j+1}+h\right\}}\right)\leq 2^{h-2}\right)\bigcap\left(d_{|m_{jb}|}\left(m_{jb},\left(\frac{sg}{r}\right)_{\left\{l_j,l_{j+1}+h\right\}}\right)\leq 2^{h-2}\right)\\&\bigcap
\left(d_{|m_{ka}|}\left(m_{ka},\left(\frac{s}{r}\right)_{\left\{l_k,l_{k+1}\right\}}\right)\leq 1\right)\bigcap\left(d_{|m_{kb}|}\left(m_{kb},\left(\frac{sg}{r}\right)_{\left\{l_k,l_{k+1}\right\}}\right)\leq 1\right)
\end{split}
\end{equation}
for all $j\in\{1,\cdots,k-1\}$ is 
\begin{align}
&\left\|\left(\mathop{\otimes}\limits_{j=1}^{k}\left(Q_{A_{j}}\otimes Q_{B_{j}}\right)\otimes I^{\otimes L}\right)\ket{\Phi}\right\|^2\\
=&\left\|\left(\mathop{\otimes}\limits_{j=1}^{k}\left(Q_{A_{j}}\otimes Q_{B_{j}}\right)\otimes I^{\otimes L}\right)\dfrac{1}{\sqrt{r}}\sum\limits_{s=0}^{r-1}\left(\mathop{\otimes}\limits_{j=1}^{k}\Ket{\psi_{t_j,\left(\frac{s}{r}\right)_{\left\{l_j,+\infty\right\}}}}_{ja}\right.\right.\left.\left.
\Ket{\psi_{t_j,\left(\frac{sg}{r}\right)_{\left\{l_j,+\infty\right\}}}}_{{jb}}\right) |u_{s}\rangle_C\right\|^2\\
=&\dfrac{1}{r}\sum\limits_{s=0}^{r-1}\left\|\left(\mathop{\otimes}\limits_{j=1}^{k}\left(Q_{A_{j}}\Ket{\psi_{t_j,\left(\frac{s}{r}\right)_{\left\{l_j,+\infty\right\}}}}_{{ja}} Q_{B_{j}}\Ket{\psi_{t_j,\left(\frac{sg}{r}\right)_{\left\{l_j,+\infty\right\}}}}_{{jb}}\right)\right)|u_{s}\rangle_C\right\|^2 ( \text{by Eq.  (\ref{us_orthonormal})})\\
\geq&\dfrac{1}{r}\sum\limits_{s=0}^{r-1}\left\|\left(\mathop{\otimes}\limits_{j=1}^{k}\left(Q_{A_{s,j}}\Ket{\psi_{t_j,\left(\frac{s}{r}\right)_{\left\{l_j,+\infty\right\}}}}_{{ja}} Q_{B_{s,j}}\Ket{\psi_{t_j,\left(\frac{sg}{r}\right)_{\left\{l_j,+\infty\right\}}}}_{{jb}}\right)\right)|u_{s}\rangle_C\right\|^2\\
=&\dfrac{1}{r}\sum\limits_{s=0}^{r-1}\left\|\mathop{\otimes}\limits_{j=1}^{k}\left(Q_{A_{s,j}}\Ket{\psi_{t_j,\left(\frac{s}{r}\right)_{\left\{l_j,+\infty\right\}}}}_{{ja}} Q_{B_{s,j}}\Ket{\psi_{t_j,\left(\frac{sg}{r}\right)_{\left\{l_j,+\infty\right\}}}}_{{jb}}\right)\right\|^2\\
\geq& \dfrac{r\cdot\left(1-\dfrac{\epsilon'}{2k}\right)^{2k}}{r}\ (\text{by Proposition \ref{phase_estimation2}})\\
\geq& 1-\epsilon'\\
>& 1-\epsilon.
\end{align}
Therefore, the lemma holds.
\end{proof}

In the following, we present a  lemma, which is relevant to proving the correctness of  Algorithm \ref{alg:CorrectAndCombine}.
\begin{Lemma}\label{CorrectStep}
Let $\omega,x$ be two $t$-bit strings ($t\geq 3$). Let $z$ be a $(h+1)$-bit string  ($2\leq h<t$). Suppose $d_t(x,\omega)\leq 1$ and $d_{h+1}\left(z,\omega_{[t-h,t]}\right)\leq 2^{h-2}$. Then there  exists a unique  $b\in\left\{0,\pm 1, \cdots,\pm 2^{h-1}\right\}$ such that 
$Sum\left(x_{[t-h,t]},b\right)=z.$
Let $b_1,b_2\in\left\{0,\pm 1, \cdots,\pm 2^{h-2}\right\}$ satisfy 
$Sum(x,b_1)=\omega$ 
and 
$Sum\left(\omega_{[t-h,t]},b_2\right)=z$,
it holds that 
$b=b_1+b_2$.
\end{Lemma}
\begin{proof}
 Since $d_t(x,\omega)\leq 2^{h-2}$ , we have $d_{h+1}\left(x_{[t-h,t]},\omega_{[t-h,t]}\right)\leq 2^{h-2}$. Then we get
\begin{equation}
d_{h+1}\left(x_{[t-h,t]},z\right)\leq d_{h+1}\left(x_{[t-h,t]},\omega_{[t-h,t]}\right)+d_{h+1}\left(\omega_{[t-h,t]},z\right) \leq 2^{h-1}. 
\end{equation}
Hence,  there  exists a unique  $b\in\left\{0,\pm 1, \cdots,\pm 2^{h-1}\right\}$ such that 
\begin{equation}
Sum\left(x_{[t-h,t]},b\right)=z.
\end{equation}

 Moreover, since
\begin{align}
Sum(x,b_1+b_2)&=Sum(Sum(x,b_1),b_2)\\
&=Sum(\omega,b_2),
\end{align}
we get
\begin{align}
Sum\left(x_{[t-h,t]},b_1+b_2\right)&=Sum\left(\omega_{[t-h,t]},b_2\right)\\
&=z.
\end{align}
Therefore, $b=b_1+b_2$. The lemma holds.
\end{proof}

In the following, we give another lemma, which is relevant for proving the correctness of  Algorithm \ref{alg:CorrectAndCombine}.
\begin{Lemma}\label{CorrectResults_correct}
Let $x,y\in\{0,1\}^t$ with $d_{t}(x,y)\leq 2^{h-2}$, where $t>2$ and $2\leq h\leq t$. Then there  exists a unique  $b_0\in\left\{0,\pm 1, \cdots,\pm 2^{h-2}\right\}$ such that $Sum(x,b_0)=y$. Then  $Sum(x,b)=y$ if and only if $Sum\left(x_{[t-h+1,t]},b\right)=y_{[t-h+1,t]}$, where $b\in\left\{0,\pm 1, \cdots,\pm 2^{h-2}\right\}$.
\end{Lemma}
\begin{proof} By Lemma \ref{d_t}, we know that there exists a unique  $b_0\in\left\{0,\pm 1, \cdots,\pm 2^{h-2}\right\}$ such that $Sum(x,b_0)=y$. In the following, we prove that $Sum(x,b)=y$ if and only if $Sum\left(x_{[t-h+1,t]},b\right)=y_{[t-h+1,t]}$, where $b\in\left\{0,\pm 1, \cdots,\pm 2^{h-2}\right\}$.  

Suppose $Sum(x,b)=y$, then we have
\begin{equation}
\left(x+b\right) \bmod 2^{t}=y,
\end{equation}
\begin{equation}
\left(x+b\right) \bmod 2^{h}=y \bmod 2^{h},
\end{equation}
that is,
\begin{equation}
\left(x_{[t-h+1,t]}+b\right) \bmod 2^h=y_{[t-h+1,t]}.
\end{equation}

Thus, we have
\begin{equation}
Sum\left(x_{[t-h+1,t]},b\right)=y_{[t-h+1,t]},
\end{equation}
where $b\in\left\{0,\pm 1, \cdots,\pm 2^{h-2}\right\}$.

Suppose $Sum\left(x_{[t-h+1,t]},b\right)=y_{[t-h+1,t]}$, where $b\in\left\{0,\pm 1, \cdots,\pm 2^{h-2}\right\}$.

 Since there  exists a unique   $b_0\in\left\{0,\pm 1, \cdots,\pm 2^{h-2}\right\}$ such that
 \begin{equation}
  Sum(x,b_0)=y,  
  \end{equation}
  $b$ is equal to $b_0$, that is, $b$ satisfies 
   \begin{equation}
  Sum(x,b)=y. 
 \end{equation}
  Consequently, the lemma holds.
\end{proof}

In the following, we give a proposition related to the correctness proof and effect analysis of  Algorithm \ref{alg:CorrectAndCombine} and the correctness proof of  Algorithm \ref{alg:DDLQA}.
\begin{Proposition}\label{prop:CorrectAndCombine}
Let $k$, $l_1$, $\cdots$, $l_{k+1}$ be postive integers with $1=l_1<\cdots<l_{k+1}=\lceil\log_2r+1\rceil+1$. Suppose $w$ is an $\lceil\log_2r+1\rceil+1$-bit string and  $x_i$ is an $(l_{i+1}+h+1-l_i)$-bit string such that 
$
d_{|x_i|}\left(x_i,w_{[l_i,l_{i+1}+h]}\right)\leq 2^{h-2}, 
$
where $i=1,\cdots,k-1$ and $2\leq h<t$.  Suppose   $x_k$ is an $(l_{k+1}+1-l_k)$-bit string such that 
$
d_{|x_k|}\left(x_k,w_{[l_k,l_{k+1}]}\right)\leq 1.  
$
Let $y=Correct(x_1,\cdots,x_k)$. Then we have
$
d_{|y|}(y,w)=d_{|x_k|}\left(x_k,w_{[l_k,l_{k+1}]}\right).
$
\end{Proposition}

\begin{proof}
Suppose we input $x_{1}, \cdots, x_k$ to Algorithm \ref{alg:CorrectAndCombine}. 
Since $d_{|x_k|}\left(x_k,w_{[l_k,l_{k+1}]}\right)\leq 1$, we have 
\begin{equation}
d_{h+1}\left((x_k)_{[1,h+1]},w_{[l_k,l_k+h]}\right)\leq 1
\end{equation}
and
\begin{equation}
d_{2}\left((x_k)_{\left[|x_k|-1,|x_k|\right]},w_{[l_{k+1}-1,l_{k+1}]}\right)\leq 1.
\end{equation}

Let $b_1,b_2\in\{0, \pm1\}$ and $b_3\in\{0, \pm1,\cdots,\pm 2^{h-2}\}$ satisfy 
\begin{equation}
Sum\left((x_k)_{\left[|x_k|-1,|x_k|\right]},b_1\right)=w_{[l_{k+1}-1,l_{k+1}]},
\end{equation}
\begin{equation}
Sum\left((x_k)_{[1,h+1]},b_2\right)=w_{[l_{k},l_{k}+h]},
\end{equation}
\begin{equation}
Sum\left((x_{k-1})_{\left[|x_{k-1}|-h,|x_{k-1}|\right]} ,b_3\right)=w_{[l_k,l_k+h]}.
\end{equation}

Let $q_{k-1}$, $p_{k-1}$, $c_{k-1}$  be the same as those in Algorithm \ref{alg:CorrectAndCombine}. 

By Lemma \ref{CorrectStep}, we have 
\begin{equation}
q_{k-1}=b_3-b_2.
\end{equation}

 Combing  Lemma  \ref{CorrectResults_correct}, we get
\begin{align}
Sum(c_{k-1},b_1)
=&Sum(p_{k-1},b_2)\circ w_{[l_k+h+1,l_{k+1}]}\\
=&Sum\left(Sum(x_{k-1},q_{k-1}),b_2\right)\circ w_{[l_k+h+1,l_{k+1}]}\\
=&Sum(x_{k-1},b_3)\circ w_{[l_k+h+1,l_{k+1}]}\\
=&w_{[l_{k-1},l_k+h]} \circ w_{[l_k+h+1,l_{k+1}]}\\
=&w_{[l_{k-1},l_{k+1}]}.
\end{align}

Hence, 
\begin{align}
d_{\left|c_{k-1}\right|}\left(c_{k-1},w_{[l_{k-1},l_{k+1}]}\right)=|b_1|=d_{\left|c_{k}\right|}\left(c_{k},w_{[l_{k},l_{k+1}]}\right).
\end{align}

Since $d_{|x_i|}\left(x_i,w_{[l_i,l_{i+1}+h]}\right)\leq 2^{h-2}$, we have 
\begin{equation}
d_{h+1}\left((x_i)_{[1,h+1]},w_{[l_i,l_i+h]}\right)\leq 2^{h-2}
\end{equation}
and
\begin{equation}
d_{h+1}\left((x_i)_{\left[|x_i|-h,|x_i|\right]},w_{[l_{i+1},l_{i+1}+h]}\right)\leq 2^{h-2},
\end{equation}
$i=1,\cdots,k-1$. 

Let $b_4,b_5\in\{0, \pm1,\cdots,\pm 2^{h-2}\}$ satisfy 
\begin{equation}
Sum\left((x_{k-1})_{[1,h+1]},b_4\right)=w_{[l_{k-1},l_{k-1}+h]},
\end{equation}
\begin{equation}
Sum\left((x_{k-2})_{\left[|x_{k-2}|-h,|x_{k-2}|\right]}, b_5\right)=w_{[l_{k-1},l_{k-1}+h]}.
\end{equation}

Let $q_j$, $p_j$, $c_j$ ($j=1,\cdots,k-2$) be the same as those in Algorithm \ref{alg:CorrectAndCombine}. 

By Lemma \ref{CorrectStep}, we have 
\begin{equation}
q_{k-2}=b_5-b_4.
\end{equation}

 Combing  Lemma  \ref{CorrectResults_correct}, we get
\begin{align}
Sum(c_{k-2},b_1)
=&Sum(p_{k-2},b_4)\circ w_{[l_{k-1}+h+1,l_{k+1}]}\\
=&Sum\left(Sum(x_{k-2},q_{k-2}),b_4\right)\circ w_{[l_{k-1}+h+1,l_{k+1}]}\\
=&Sum(x_{k-2},b_5)\circ w_{[l_{k-1}+h+1,l_{k+1}]}\\
=&w_{[l_{k-2},l_{k-1}+h]} \circ w_{[l_{k-1}+h+1,l_{k+1}]}\\
=&w_{[l_{k-2},l_{k+1}]}.
\end{align}

Hence, 
\begin{align}
d_{\left|c_{k-2}\right|}\left(c_{k-2},w_{[l_{k-2},l_{k+1}]}\right)=|b_1|=d_{\left|c_{k-1}\right|}\left(c_{k-1},w_{[l_{k-1},l_{k+1}]}\right).
\end{align}

By induction, it can be proven that
\begin{align}
\hspace{-2.5em}d_{\left|c_{1}\right|}\left(c_{1},w_{[1,l_{k+1}]}\right)=
d_{\left|c_{k}\right|}\left(c_{k},w_{[l_{k},l_{k+1}]}\right).
\end{align}

Since $y=c_1, c_k=x_k$ and $l_{k+1}=\lceil\log_2r+1\rceil+1$, the proposition holds.
\end{proof}

In the following, we give a theorem, which proves that  Algorithm \ref{alg:DDLQA} is correct.

\begin{Theorem}\label{CorrectnessForDOFA}
In step 8 of Algorithm \ref{alg:DDLQA},  the probability that $m_a$ and $m_b$ satisfy
\begin{equation}
\left(\left|\dfrac{m_a}{2^{|m_a|}}-\dfrac{s}{r}\right|\leq2^{-\lceil\log_2r+1\rceil}\right)\bigcap \left(\left|\dfrac{m_b}{2^{|m_b|}}-\dfrac{sg(\bmod\ r)}{r}\right|\leq2^{-\lceil\log_2r+1\rceil}\right)
\end{equation}
 is greater than $1-\epsilon$, where $s\in \{0,1,\cdots,r-1\}$.

In step 9 of Algorithm \ref{alg:DDLQA}, the probability that $b\equiv a^{\hat{g}} (\bmod\ N)$ holds is  greater than $\dfrac{r-1}{r}(1-\epsilon)$.
\end{Theorem}
\begin{proof}
In step 8 of Algorithm \ref{alg:DDLQA}, based on Lemma \ref{distancestring} and Proposition \ref{prop:CorrectAndCombine}, we get that the probability of $m_a$ and $m_b$ satisfying
\begin{equation}\label{eqdistanceprobability1}
\left(d_{|m_{a}|}\left(m_{a},\left(\frac{s}{r}\right)_{\{1,|m_a|\}}\right)\leq 1\right) \bigcap
\left(d_{|m_{b}|}\left(m_{b},\left(\frac{sg}{r}\right)_{\{1,|m_b|\}}\right)\leq 1\right)
\end{equation}
is greater than $1-\epsilon$, where $s\in \{0,1,\cdots,r-1\}$.

If $\dfrac{s}{r}$ is an integer, i.e., $s=0$, then we have 
\begin{equation}\label{eqdistanceprobability2}
m_{a}=\left(\frac{s}{r}\right)_{\{1,|m_a|\}}.
\end{equation}

If $\dfrac{s}{r}$ is not an integer, then we have
\begin{equation}
2^{-{\lceil\log_2r+1\rceil-1}}<\dfrac{1}{r}\leq \dfrac{s }{r}\leq \dfrac{r-1}{r}<1-2^{-{\lceil\log_2r+1\rceil-1}}.
\end{equation}
So we get $\left(\dfrac{s}{r}\right)_{\{1,|m_a|\}}=\left(\dfrac{s}{r}\right)_{\{1,\lceil\log_2r+1\rceil+1\}}$ is not $00\ldots0$ or $11\ldots1$. 

Thus, we have
\begin{equation}
\left|m_a-\left(\frac{s}{r}\right)_{\{1,|m_a|\}}\right|=d_{|m_{a}|}\left(m_{a},\left(\frac{s}{r}\right)_{\{1,|m_a|\}}\right)\leq 1.
\end{equation}

Since $\left|2^{|m_a|}\cdot\dfrac{s}{r}-\left(\dfrac{s}{r}\right)_{\{1,|m_a|\}}\right|\leq 1$, we have
\begin{align}
&\left|m_a-2^{|m_a|}\cdot\dfrac{s}{r}\right|\\\leq& \left|m_a-\left(\frac{s}{r}\right)_{\{1,|m_a|\}}\right|+\left|\left(\frac{s}{r}\right)_{\{1,|m_a|\}}-2^{|m_a|}\cdot\dfrac{s}{r}\right|\\
\leq& 2.
\end{align}

Therefore, we have

\begin{equation}
\left|\dfrac{m_a}{2^{|m_a|}}-\dfrac{s}{r}\right|\leq2^{1-|m_a|}=2^{-\lceil\log_2r+1\rceil}.
\end{equation}

Similarly, we can get
\begin{equation}
\left|\dfrac{m_b}{2^{|m_b|}}-\dfrac{sg(\bmod\ r)}{r}\right|\leq2^{-\lceil\log_2r+1\rceil}.
\end{equation}

Therefore, in step 8 of Algorithm \ref{alg:DDLQA},  the probability of $m_a$ and $m_b$ satisfying
\begin{equation}\label{eqmamb}
\left(\left|\dfrac{m_a}{2^{|m_a|}}-\dfrac{s}{r}\right|\leq2^{-\lceil\log_2r+1\rceil}\right)\bigcap \left(\left|\dfrac{m_b}{2^{|m_b|}}-\dfrac{sg(\bmod\ r)}{r}\right|\leq2^{-\lceil\log_2r+1\rceil}\right)
\end{equation}
 is greater than $1-\epsilon$, where $s\in \{0,1,\cdots,r-1\}$.

 From Eq. (\ref{eqmamb}) and $r$ being a prime number greater than 2, we can get the probability of 
 \begin{equation}\label{eqmambr}
\left(\left|\dfrac{m_ar}{2^{|m_a|}}-s\right|<2^{-1}\right)\bigcap \left(\left|\dfrac{m_br}{2^{|m_b|}}-sg(\bmod\ r)\right|<2^{-1}\right)
\end{equation}
 is greater than $1-\epsilon$, where $s\in \{0,1,\cdots,r-1\}$.
 
Since $\hat{m}_a$ and $\hat{m}_b$ are rounded by $m_ar/2^{|m_a|}$ and $m_br/2^{|m_b|}$, respectively, there are $\hat{m}_a=s$ and $\hat{m}_b=sg(\bmod\ r)$.

Due to $r$ being a prime, if $s\neq 0$, it follows that the multiplicative inverse of $s (\bmod\ r)$ exists,  and we have 
 \begin{equation}\label{d}
 \hat{g}=\hat{m}_a^{-1}\hat{m}_b(\bmod\ r)=s^{-1}(sg(\bmod\ r))(\bmod\ r)=g.
\end{equation}

The probability that $s$ is taken from $\{0,1,\cdots,r-1\}$ and $s\neq 0$ is $\dfrac{r-1}{r}$.  Then according to Eq. (\ref{eqmambr}) and Eq. (\ref{d}), the probability that $b\equiv a^{\hat{g}} (\bmod\ N)$ holds is  greater than $\dfrac{r-1}{r}(1-\epsilon)$.
\end{proof}

\subsection{Algorithm  complexity analysis}\label{DDLQA Algorithm complexity  analysis}

In the following, we analyze the complexity of Algorithm \ref{alg:DLQA} and Algorithm \ref{alg:DDLQA}. 
The complexity of the circuit of (distributed) quantum algorithm for discrete logarithm problem depends on the construction of $C_t(M_a)$ ($C_t(M_b)$).  
The implementation of $C_t(M_a)$ ($C_t(M_b)$) proposed by Shor has a time complexity of $O\left(L^3\right)$ and a space complexity of $O(L)$ \cite{shor1994algorithms}.

\textbf{Space complexity.} The implementation of the operator $C_t(M_a)$ ($C_t(M_b)$)  needs $t+L$ qubits plus $c$ auxiliary qubits for any positive integer $a$ $(b)$, where $c$ is $L+O(1)$. 
As a result, the maximum number of qubits required by a
single node of Algorithm \ref{alg:DLQA} is $2\Bigg(\lceil\log_2r+1\rceil$+$\left.\left\lceil\log_2\left(2+\dfrac{1}{\epsilon}\right)\right\rceil\right)+L+O(1)$.
The maximum number of qubits required by a
single node of Algorithm \ref{alg:DDLQA}  is $2\left(\dfrac{\lceil\log_2r+1\rceil+1}{k}+\right.$ $\left.\left\lceil\log_2\left(2+\dfrac{k}{\epsilon'}\right)\right\rceil\right)+L+O(1)$, where $k$ is the number of compute nodes, $0<\epsilon'<\epsilon<1$, and $L=\left\lfloor\log_2 N\right\rfloor+1$. 

\textbf{Time complexity.} The operator $C_t(M_a)$ ($C_t(M_b)$) can be implemented by  $O\left(tL^2\right)$ elementary gates. Thus,  the   gate complexity (or time complexity) in both Algorithm \ref{alg:DLQA}  and  Algorithm \ref{alg:DDLQA} is $O\left(L^3\right)$.

\textbf{Circuit depth.} By Fig. \ref{fig:controlled-U}, we know that the circuit depth of $C_t(M_a)$ $(C_t(M_b))$ depends on the circuit depth of  controlled-$M_a^{2^n}$ $\Big($controlled-$M_b^{2^n}$$\Big)$ $(n=0,1,\cdots,t-1)$ and $t$. The circuit depth of controlled-$M_a^{2^n}$ $\Big($controlled-$M_b^{2^n}$$\Big)$ is $O(L^2)$.
 By  observing the value ``$t$" in Algorithm \ref{alg:DLQA}  and Algorithm \ref{alg:DDLQA}, we clearly get that the circuit depth of each node in  Algorithm \ref{alg:DDLQA}  is less than Algorithm \ref{alg:DLQA}, even though both are $O(L^3)$.

\textbf{Quantum communication complexity.} In the following, we analyze the  quantum communication complexity of Algorithm \ref{alg:DDLQA}. First, Node $T_1$ transmits $L$ qubits from register $C$ to Node $T_2$. Then, Node $T_2$ transmits $L$ qubits from register $C$ to Node $T_3$, and so on.  Finally, Node $T_{k-1}$ transmits $L$ qubits from register $C$ to Node $T_k$.  
 Therefore, the quantum communication complexity of Algorithm \ref{alg:DDLQA} is $O(kL)$. 
 
 \textbf{Probability of success.} The success probability of Algorithm \ref{alg:DLQA} is $\dfrac{r-1}{r}(1-\epsilon)$, while the success probability of Algorithm \ref{alg:DDLQA} is $\dfrac{r-1}{r}(1-\epsilon')$. Since $\epsilon'<\epsilon$, the success probability of Algorithm \ref{alg:DDLQA} is higher than that of Algorithm \ref{alg:DLQA}.

\begin{table*}[h]
\renewcommand{\arraystretch}{2.5}
\caption{Comparisons of Algorithm \ref{alg:DLQA} and Algorithm \ref{alg:DDLQA}.} \label{table:DDLQA}
\resizebox{\linewidth}{!}{ 
\begin{tabular}{*{6}{c}}
\toprule
Algorithm & Space complexity &   Time complexity &Circuit depth & \makecell[c]{Quantum communication \\ complexity}&Probability of success\\
\midrule
Algorithm \ref{alg:DLQA} & \makecell[c]{$2\left(\lceil\log_2r+1\rceil+\left\lceil\log_2\left(2+\dfrac{1}{\epsilon}\right)\right\rceil\right)+L+O(1)$}  &$O\left(L^3\right)$  & $O\left(L^3\right)$  & $0$ & $\dfrac{r-1}{r}(1-\epsilon)$ \\
\hline
Algorithm \ref{alg:DDLQA} & \makecell[c]{$2\left(\dfrac{\lceil\log_2r+1\rceil+1}{k}+\left\lceil\log_2\left(2+\dfrac{k}{\epsilon'}\right)\right\rceil\right)+L+O(1)$} & $O\left(L^3\right)$ & $O\left(L^3\right)$ & $O(kL)$ & $\dfrac{r-1}{r}(1-\epsilon')$ \\ 
\bottomrule
\end{tabular}
}
\end{table*}

\section{Conclusions} \label{sec:conclusions}

By employing the technical method discovered by Xiao and Qiu et al \cite{Xiao2023DQAkShor},  we have proposed a  quantum-classical hybrid distributed quantum algorithm for discrete logarithm problem, where  quantum  algorithm is used to obtain results, while classical algorithm is used to ensure  the accuracy of the results. There are multiple computers working sequentially via quantum teleportation. Each of them can obtain an estimation of partial bits of the result with high probability.   The space complexity of our algorithm is lower than Shor's quantum algorithm. 
Also, the success probability of  our algorithm is higher than that of Shor's quantum algorithm.
In addition, we have generalized the classical error correction technique proposed in \cite{Xiao2023DQAkShor}, by extending  it from three bits to more than three bits.

We have shown that our  algorithm  has advantages over the  quantum algorithm for discrete logarithm problem in space complexity, circuit depth, and probability of success, which may be crucial  in the NISQ era. In  future research, we may explore the applications of our algorithm to lattice-based cryptanalysis, and design a  distributed quantum algorithm for solving systems of linear equations.


\end{document}